\documentclass[11pt]{article}

\usepackage{enumerate}
\usepackage{amsmath}
\usepackage{amssymb,latexsym}
\usepackage{amsthm}
\usepackage{color}
\usepackage{cancel}
\usepackage{graphicx}

\usepackage{tikz}

\usepackage{hyperref}

\usepackage{amscd}

\DeclareMathOperator{\C}{\mathcal{C}}

\newtheorem{theorem}{Theorem}[section]

\newtheorem{lemma}[theorem]{Lemma}
\newtheorem{corollary}[theorem]{Corollary}
\newtheorem{definition}[theorem]{Definition}
\newtheorem{proposition}[theorem]{Proposition}

\newtheorem{remark}[theorem]{Remark}

\newcommand{\fqn}{\mathbb{F}_{q^n}}

\newcommand{\cG}{{\mathcal G}}

\newcommand{\cH}{{\mathcal H}}
\newcommand{\cD}{{\mathcal D}}

\newcommand{\F}{{\mathbb F}}
\newcommand{\Z}{{\mathbb Z}}

\newcommand{\T}{\mathrm{Tr}}

\newcommand{\fq}{{\mathbb F}_{q}}

\newcommand{\la}{\langle}
\newcommand{\ra}{\rangle}

\newcommand{\PG}{\mathrm{PG}}
\newcommand{\N}{\mathrm{N}}

\usepackage{lineno,xcolor}
\title{On the list decodability of rank-metric codes containing Gabidulin codes}

\author{Paolo Santonastaso and Ferdinando Zullo}
\date{}

\begin{document}
\maketitle

\begin{abstract}
Wachter-Zeh in \cite{wachter-zhe_2013}, and later together with Raviv \cite{raviv_2016}, proved that Gabidulin codes cannot be efficiently list decoded for any radius $\tau$, providing that $\tau$ is large enough. Also, they proved that there are infinitely many choices of the parameters for which Gabidulin codes cannot be efficiently list decoded at all.
Subsequently, in \cite{TrombZullo} these results have been extended to the family of generalized Gabidulin codes and to further family of MRD-codes.
In this paper, we provide bounds on the list size of rank-metric codes containing generalized Gabidulin codes in order to determine whether or not a polynomial-time list decoding algorithm exists. 
We detect several families of rank-metric codes containing a generalized Gabidulin code as subcode which cannot be efficiently list decoded for any radius large enough and families of rank-metric codes which cannot be efficiently list decoded.
These results suggest that rank-metric codes which are $\F_{q^m}$-linear or that contains a (power of) generalized Gabidulin code cannot be efficiently list decoded for large values of the radius.
\end{abstract}

\bigskip
{\it AMS subject classification:} 	94B35; 94B05.

\bigskip
{\it Keywords:} rank-metric code; list decoding; linearized polynomial; Gabidulin code.

\section{Introduction}

Rank-metric codes were introduced by Delsarte \cite{delsarte_bilinear_1978} in 1978 and they have been intensively investigated in recent years because of their applications; we refer to \cite{sheekey_newest_preprint} for a recent survey on this topic.
The set of $m \times n$ matrices $\F_q^{m\times n}$ over $\F_q$ may be equipped with the \emph{rank metric}, defined by
\[d(A,B) = \mathrm{rk}\,(A-B).\]
A subset $C \subseteq \F_q^{m\times n}$ endowed with the rank metric is called a \emph{rank-metric code} (shortly, a \emph{RM}-code).
The minimum distance of $C$ is defined as
\[d = \min\{ d(A,B) \colon A,B \in C,\,\, A\neq B \}.\]
Delsarte showed in \cite{delsarte_bilinear_1978} that the parameters of these codes must obey a Singleton-like bound, i.e.
\begin{equation}\label{eq:Singleton} |C| \leq q^{\max\{m,n\}(\min\{m,n\}-d+1)}. 
\end{equation}
When equality holds, we call $C$ a \emph{maximum rank distance} (\emph{MRD} for short) code.
Examples of $\F_q$-linear MRD-codes were first found in \cite{delsarte_bilinear_1978,gabidulin_MRD_1985}, often called \emph{Gabidulin codes}.

More recently in \cite{sheekey_new_2016}, the author exhibited two infinite families of linear MRD-codes which are not equivalent to generalized Gabidulin codes. We call them {\it twisted Gabidulin codes} and {\it generalized twisted Gabidulin} codes. In \cite{lunardon_generalized_2015} it was shown that the latter family contains both generalized Gabidulin codes and twisted Gabidulin codes as proper subsets.
From then, more constructions arises, see e.g.\ \cite{BZZ,CMPZ,CsMPZh,CsMZ2018,LZ2,MMZ,OzbudakOtal,SheekeyLondon,trombetti_new}.

\medskip

From now on suppose $n \leq m$. 

There are other equivalent ways of representing a rank-metric code of $\F_q^{m\times n}$. For our purpose, we will see such codes also as subsets of $\F_{q^m}^n$.

For a vector $\mathbf{v}=(v_1,\ldots,v_n)\in \F_{q^m}^n$, we define its rank weight as follows
\[ \mathrm{rk}(\mathbf{v})=\dim_{\F_q}  \la v_1,\ldots, v_n \ra_{\F_q}. \]
The rank distance between two vectors $\mathbf{u}, \mathbf{v} \in \F_{q^m}^n$ is defined as $d(\mathbf{u},\mathbf{v})=\mathrm{rk}(\mathbf{u}-\mathbf{v})$.
A rank-metric code of $\F_{q^m}^n$ is a subset of $\F_{q^m}^n$ equipped with the aforementioned metric.
The same bound \eqref{eq:Singleton} holds and hence we can define again an MRD-code $\mathrm{C}$ as the code whose parameters attain the equality in \eqref{eq:Singleton}, i.e. $|\mathrm{C}|=q^{mk}$ and for each $\mathbf{u},\mathbf{v}\in \mathrm{C}$ with $\mathbf{u}\neq\mathbf{v}$ we have that $\mathrm{rk}(\mathbf{u}-\mathbf{v})\geq n-k+1$.
Recall also the following definition. For each element $\mathbf{w} \in \F_{q^m}^n$ and $\tau \in \Z^+$, we define
$$B_{\tau}(\mathbf{w}):= \{\mathbf{c} \in \F_{q^m}^n  \colon  \mathrm{rk}(\mathbf{w}-\mathbf{c}) \leq \tau\}.$$

MRD-codes with efficient decoding algorithm are of great interest in practice. Several decoding algorithms exist for Gabidulin codes, see \cite{gabidulin_MRD_1985,Loi06,PWZ,Random,RP04,WAS11,WAS13,WSB10} and for further MRD-codes \cite{Li,Li2,Randrianarisoa,RandrianarisoaRosenthal}.
Also, the list decoding problem for rank-metric codes has been intensively investigated, as it allows to handle a greater number of errors.
The problem of list decoding was originally introduced in the Hamming setting by Elias in \cite{Elias} and Wozencraft in \cite{Wozencraft} and can be stated in a very general fashion. 
Let $\mathrm{C}$ be any rank-metric code of length $n$ in  $\F_{q^m}^n$, let $\tau$ be a positive integer and given a received word, output the list of all codewords of the code within distance $\tau$ from it.   

A \emph{list decoding algorithm} returns the list of all codewords with distance at most $\tau$ from any given word.
Consequently, we say that $\mathrm{C}$ is \emph{efficiently list decodable} at the radius $\tau$, if there exists a polynomial-time list decoding algorithm. Of course, if there exists a word $\mathbf{w} \in \F_{q^m}^n \setminus \mathrm{C}$ for which $B_{\tau}(\mathbf{w})\cap \mathrm{C}$ has exponential size in the length of the code, such an algorithm cannot exist since writing down the list already has exponential complexity. When such an algorithm does not exist we say that $\mathrm{C}$ is \emph{not efficiently list decodable at the radius} $\tau$.
Clearly, all the rank-metric codes are list decodable for any radius less than or equal to the corresponding unique decoding radius. 
Furthermore, if $\mathrm{C}$ is not efficiently list decodable at the radius greater than its unique decoding radius, we say that $\mathrm{C}$ is \emph{not efficiently list decodable at all}.
See \cite{guruswami} for further details on the list decodability issue.

Several construction of rank-metric codes which can be efficiently list decoded from a larger number of errors are known, see e.g.\ \cite{guruswami_2016}. As far as we can tell, none of them are either $\F_{q^m}$-linear or MRD.
Wachter-Zeh in \cite{wachter-zhe_2013}, and later with Raviv \cite{raviv_2016}, investigated the list decoding problem for Gabidulin codes. 

In \cite[Theorem 1]{wachter-zhe_2013}, Wachter-Zeh proved that Gabidulin codes in $\F_{q^m}^n$ with minimum distance $d$ cannot be efficiently list decoded at any radius $\tau$ such that \begin{equation*}\label{eq:boundontau} \tau \geq \frac{m+n}{2}-\sqrt{\frac{{(m+n)}^2}{2} -m(d-\epsilon)},\end{equation*} where $0\leq \epsilon  < 1$.

In \cite{raviv_2016} the authors improved this result under specific restrictions for the involved parameters. As a consequence, they showed infinite families of Gabidulin codes which are not efficiently list decodable at all. 
Such results have been adapted for a more general family of MRD-codes in \cite{TrombZullo}.
In view of these results, Renner, Puchinger and Wachter-Zeh in \cite{LIGA} proposed a new cryptosystem, called LIGA, based on the hardeness of list decoding (and interleaved decoding) of Gabidulin codes.
Motivated by such results we explore negative results on the list decoding of rank-metric codes, using the approach developed in \cite{wachter-zhe_2013} and in \cite{raviv_2016,TrombZullo}. 
Indeed, as shown in \cite{wachter-zhe_2013}, if $n\mid m$, then to ensure the existence of a ball not centered in a codeword with a fixed radius and meeting the code in a enough large number of codewords, may be translated in to find an enough large class of linearized polynomials with maximum kernel defined over $\F_{q^n}$.
This methods applies when the rank-metric code analyzed contains a Gabidulin code, indeed we are able to detect new families of rank-metric codes which are not efficiently list decoded either at all or from a larger number of errors.

The paper is organized as follows.
Section \ref{sec:prel} is devoted to definitions and results on linearized polynomials and rank-metric codes.
In Section \ref{sec:subpol} we give the definition of $\sigma$-subspace polynomial, which generalizes the concept of subspace polynomial, and we describe properties and known families of such polynomials, whereas in Subsection \ref{sec:gentrace} we extend the subspace polynomials introduced in \cite{HKP} obtaining new families of $\sigma$-subspace polynomials.
In Section \ref{sec:gen}, applying the techniques of \cite{wachter-zhe_2013} and of \cite{raviv_2016,TrombZullo}, we first obtain a Johnson-like bound for rank-metric codes containing generalized Gabidulin codes (see \eqref{eq:Johnson}) and then we give a bound on the list size of rank-metric codes containing generalized Gabidulin codes, which relies on the existence of an enough large family of $\sigma$-subspace polynomials. 
In Section \ref{sec:appl1}, we apply the machinery developed in Section \ref{sec:gen} by using the families of $\sigma$-subspace polynomials.
As a byproduct, in Section \ref{sec:appl2} we are able to exhibit families of rank-metric codes which are not efficiently list decoded either at all or from a larger number of errors, some of them are also MRD-codes or close to be MRD.
These results have a natural applications to the list decodability of subspace codes, which we describe in Section \ref{sec:subspacecodes}.
Finally, we conclude the paper by listing some open problems, see Section \ref{sec:probl}.

\section{Preliminaries}\label{sec:prel}

Through this paper, $q$ is a power of a prime $p$, $\F_q$ denotes the finite field of order $q$, $\F_{q^n}$ denotes its finite extension of degree $n\geq 2$.
For any positive divisor $r$ of $n$, let $\mathrm{Tr}_{q^n/q^r}(x)=\sum_{i=0}^{n/r-1} x^{q^{ir}}$ and $\mathrm{N}_{q^n/q^r}(x)=x^{\frac{q^n-1}{q^r-1}}$.

\subsection{Linearized polynomials}

Let $\sigma$ be a generator of $\mathrm{Gal}(\F_{q^m}\colon \F_q)$.
The set of $\sigma$-polynomials (or $\sigma$-linearized polynomials) with coefficients in $\mathbb{F}_{q^m}$ is
\[\mathcal{L}=\left\{\sum_{i=0}^{k} a_i x^{\sigma^i}: a_i\in \mathbb{F}_{q^m}  , k \in \mathbb{N}_0 \right\}.\]
Any polynomial $f$ in $\mathcal{L}$ gives rise to an $\mathbb{F}_q$-linear map $x\in \mathbb{F}_{q^m} \mapsto f(x) \in \mathbb{F}_{q^m}$. If $a_k \neq 0$ we will refer to $k$ as to the $\sigma$-degree of $f$, which will be denoted by $\deg_{\sigma}(f)$.
It is well known that $(\mathcal{L} / (x^{\sigma^m}-x),+,\circ,\cdot)$, where $+$ is the addition of maps, $\circ$ is the composition of maps modulo $x^{\sigma^m}-x$ and $\cdot$ is the scalar multiplication by elements of $\mathbb{F}_q$, is isomorphic to the algebra of $m\times m$ matrices over $\mathbb{F}_q,$ and hence to $\mathrm{End}_{\mathbb{F}_q}(\mathbb{F}_{q^m})$; i.e., the set of endomorphisms on $\mathbb{F}_{q^m}$ seen as an $\mathbb{F}_q$-algebra. In the following we will denote this algebra by $\mathcal{L}_{m,\sigma}$ (and by $\mathcal{L}_{m,q}$ if $\sigma$ coincides with $x\in \F_{q^m}\mapsto x^q\in \F_{q^m}$) and we will always identify the elements of $\mathcal{L}_{m,\sigma}$ with the endomorphisms of $\mathbb{F}_{q^m}$ they represent. 
Consequently, we will speak also of \emph{kernel} and \emph{rank} of a polynomial meaning by this the kernel and rank of the corresponding endomorphism.
Clearly, the kernel of $f\in \mathcal{L}_{m,\sigma}$ coincides with the set of the roots of $f$ over $\mathbb{F}_{q^m}$ and as usual $\dim_{\mathbb{F}_q} \mathrm{Im}(f)+\dim_{\mathbb{F}_q} \ker(f)=m$.

Consider the non-degenerate symmetric bilinear form of $\F_{q^m}$ over $\F_q$ defined by $\langle x, y \rangle := \T_{q^m/q}(xy)$, for every $x,y \in \F_{q^m}$. Then the \emph{adjoint} $\hat{f}$ of the $q$-polynomial $f(x)=\sum_{i=0}^{m-1}a_ix^{q^i} \in \mathcal{L}_{m,q}$ with respect to the bilinear form $\langle,\rangle$ is 

$$\hat{f}(x)=\sum_{i=0}^{m-1}a_i^{q^{m-i}}x^{q^{m-i}},$$
i.e.
$$\T_{q^m/q}(f(x)y)=\T_{q^m/q}(x\hat{f}(y)).$$

In \cite[Lemma 2.6]{BGMP}, the authors proved that
$$\left\{\frac{f(x)}{x}:x \in \F_{q^m}^*\right\}=\left\{\frac{\hat{f}(x)}{x}:x \in \F_{q^m}^*\right\}.$$ 
From this equality follows that
\begin{equation}\label{eq:dimkeradj}
    \dim_{\F_q}(\ker(f))=\dim_{\F_q}(\ker(\hat{f})).
\end{equation}

Therefore, the adjoint of a $\sigma$-linearized polynomial $f(x)=\sum_{i=0}^{m-1}a_ix^{\sigma^i}$ is $\hat{f}(x)=\sum_{i=0}^{m-1}a_i^{\sigma^{m-i}}x^{\sigma^{m-i}}$, and \eqref{eq:dimkeradj} still holds.

\begin{remark} \label{rem:adjoint}
Let $f(x)=\sum_{i=0}^{k}a_i x^{\sigma^i}\in \mathcal{L}_{m,\sigma}$ with $a_k \ne 0$. Then, by \eqref{eq:dimkeradj},
$g(x)=(\hat{f}(x))^{\sigma^k}=\sum_{i=0}^k (a_{k-i}x)^{\sigma^i}$ is a $\sigma$-polynomial and $\dim_{\fq}(\ker(g))=\dim_{\F_q}(\ker(f))$.
\end{remark}

\subsection{Rank-metric codes}

It is convenient to represent rank-metric codes as subsets of $\mathcal{L}_{m,\sigma}$, see e.g.\ \cite{sheekey_newest_preprint}.
Let recall a large class of rank-metric codes.
Let $f_1$ and $f_2$ be two additive functions of $\F_{q^m}$ such that $|\mathrm{Im}(f_1)\times\mathrm{Im}(f_2)|=q^m$ and let $k \leq m-1$.

Then the set 

\begin{small}
\begin{equation}\label{eq:codes}
\cH_{m,k,{\sigma}}(f_1,f_2)=\left\{f_1(a)x + \sum_{i=1}^{k-1}a_ix^{\sigma^i}+f_2(a)x^{\sigma^k} \colon a,a_i \in \F_{q^m}\right\}
\end{equation}
\end{small}
is an $\F_p$-linear rank-metric code in $\mathcal{L}_{m,\sigma}$ of size $q^{mk}$.
Furthermore, if $f_1$ and $f_2$ are such that $\N_{q^m/q}(f_1(a)) \neq (-1)^{mk}\N_{q^m/q}(f_2(a))$ for every $a\in \F_{q^m}^*$, by \cite[Proposition 1]{sheekey_newest_preprint}, $\cH_{m,k,{\sigma}}(f_1,f_2)$ defines an MRD code in $\mathcal{L}_{m,\sigma}$ with minimum distance $d=m-k+1$. For instance, if $f_1(a)=a$ and $f_2(a)=0$ for each $a \in \F_{q^m}$, then $\cH_{m,k,\sigma}(f_1,f_2)$ is a generalized Gabidulin code (commonly indicated with the symbol $\cG_{m,k,\sigma}$).

In the following table we summarize the examples of the known MRD-codes of $\mathcal{L}_{m,\sigma}$ that can be represented as in (\ref{eq:codes}).

\begin{table*}[ht]\label{table:ex}
\caption{Known examples of MRD-codes in $\mathcal{L}_{m,\sigma}$}
\centering
\begin{tabular}{c c c c c c c}
\hline\hline
\text{Symbol} & $\sigma$ & $f_1(a)$ & $f_2(a)$ & \text{Conditions} & \text{References} \\
\hline
$\cG_{m,k}$ & $q$      &  $a$           & 0             & --- & \cite{delsarte_bilinear_1978,gabidulin_MRD_1985} \\
$\cG_{m,k,\sigma}$ & $q^s$      &  $a$           & 0             & --- & \cite{kshevetskiy_new_2005} \\
$\cH_{m,k,\sigma}(\eta,h)$ &  $q^s$     &  $a$           & $\eta a^{q^h}$  & $N_{q^m/q}(\eta) \neq (-1)^{mk}$ & \cite{sheekey_new_2016,lunardon_generalized_2015} \\
$\overline{\cH}_{m,k,\sigma}(\eta,h)$ &  $q^s$     &  $a$           & $\eta a^{p^h}$  & $N_{q^m/p}(\eta) \neq (-1)^{mk}$ & \cite{OzbudakOtal} \\
$\cD_{m,k,\sigma}(\eta)$    &  $q^s$     &   $a+a^{q^{m/2}}$     &    $\eta (b+b^{q^{m/2}})$ & $m$ \text{even}, $N_{q^m/q}(\eta) \notin \square$, $a,b \in \F_{q^{m}}$ & \cite{trombetti_new} \\
\hline
\end{tabular}
\\ where $\N_{q^m/p}(\eta)=\eta^{1+p+\ldots+p^{m\ell-1}}$ ($q=p^\ell$) and $\square$ is the set of square elements in $\F_{q^m}$.
\end{table*}

Except for generalized Gabidulin codes, all of the examples above have {\it Gabidulin index} $k-1$ (see \cite{GZ}), that is they contain a subcode equivalent to $\cG_{m,k-1,\sigma}$.

Any rank-metric code of $\mathcal{L}_{m,\sigma}$ defines a rank-metric code in $\F_{q^m}^n$, in view of the following result.

\begin{lemma}\label{lemma:rankmetricvect}
Let $\mathcal{C}$ be a RM-code of $\mathcal{L}_{m,\sigma}$ with $\lvert \mathcal{C} \rvert =q^{mk}$ and $d(\mathcal{C})=m-h$, with $h$ a positive integer such that $h \geq k-1$. Let $n$ be a positive integer greater than or equal to $k$
and let $S=\{\alpha_1,\ldots,\alpha_n\}$ be a set of $n$ distinct $\F_q$-linearly
independent elements of $\F_{q^m}$. 
Then the rank-metric code
$$\mathrm{C}=\{ (g(\alpha_1),\ldots, g(\alpha_n)): g \in \mathcal{C}  \}\subseteq \F_{q^m}^n$$
is a code of $\F_{q^m}^n$ with $\lvert \mathrm{C} \rvert=q^{mk}$ and $n-h\leq d(\mathrm{C}) \leq n-k+1$. In particular, if $\mathcal{C}$ is an MRD-code then $\mathrm{C}$ is an MRD-code, that is $d(\mathrm{C})=n-k+1$
\end{lemma}
\begin{proof}
Consider $U_S$ as the $\F_q$-subspace of $\F_{q^m}$ spanned by the elements of $S$.

Let $f \in \mathcal{C}$. Since $d(\C)=m-h$, then $m-h \leq \mathrm{rk}(f)=m-\dim_{\F_q}(\ker(f))$, so that $\dim_{\F_q} (\ker(f)) \leq h$. 
Now, consider
$$\phi: f \in \mathcal{C} \mapsto (f(\alpha_1),\ldots, f(\alpha_n)) \in \mathrm{C}.$$
It is easy to see that $\phi$ is a bijection, so that $\lvert \mathrm{C} \rvert =\lvert \mathcal{C} \rvert$. The Singleton bound then implies $d(\mathrm{C})\leq n-k+1$. Also,
\[ \mathrm{rk}(\phi(f))=n-\dim_{\F_q}(\ker(f)\cap U_S)\geq n-h, \]
that is $d(\mathrm{C}) \geq n-h$.
\end{proof}

\begin{remark}
If $f \in \mathcal{L}_{m,\sigma}$ and let $\alpha_1,\ldots,\alpha_n \in \F_{q^m}$ as in Lemma \ref{lemma:rankmetricvect}, then we denote by $c_f$ the evaluation vector $(f(\alpha_1),\ldots,f(\alpha_n))\in \F_{q^m}^n$.
Clearly,
\[ \mathrm{rk}(c_f)=n-\dim_{\F_q} (\ker\,(f)\cap \langle \alpha_1,\ldots,\alpha_n \rangle_{\F_q}). \]
\end{remark}

\section{Subspace polynomials}\label{sec:subpol}

Let $m\geq 2$ and $\sigma$ be a generator of $\mathrm{Gal}(\mathbb{F}_{q^m}\colon \mathbb{F}_q)$.

\smallskip

The number of roots of a $\sigma$-polynomial is bounded by its $\sigma$-degree, by \cite[Theorem 5]{GQ2009}.
Precisely, in \cite{teoremone} $\sigma$-polynomials over finite fields for which the dimension of the kernel coincides with their $\sigma$-degree are called $\sigma$-\emph{polynomials with maximum kernel}. Following \cite{TrombZullo} and extending the notion of subspace polynomial, we will call a \textbf{monic} $\sigma$-polynomial a $\sigma$-\emph{subspace polynomial} if it is a $\sigma$-polynomial with maximum kernel.
The reason of such a name relies on the following property, which is well known when $\sigma$ coincides with $x\in \F_{q^m}\mapsto x^q\in \F_{q^m}$.

\begin{proposition}
Let $r$ be a positive integer with $0 \leq r \leq m-1$. The set $\mathcal{P}_{\sigma,r}$ of $\sigma$-subspace polynomials having $\sigma$-degree $r$ is in one-to-one correspondence with the set of $r$-dimensional $\mathbb{F}_q$-subspaces of $\mathbb{F}_{q^m}$.
In particular, its size is ${m\brack r}_q$.
\end{proposition}
\begin{proof}
The kernel of a $\sigma$-subspace polynomial having $\sigma$-degree $r$ is an $r$-dimensional subspace of $\mathbb{F}_{q^m}$.
Also, let $U$ be an $r$-dimensional $\mathbb{F}_q$-subspace and let $u_1,\ldots,u_r$ an its $\mathbb{F}_q$-basis.
Then
{\footnotesize
\[s(x)=(-1)^{r+1}\left| \begin{pmatrix}
u_1 & u_1^\sigma & \cdots & u_1^{\sigma^{r-1}} \\
\vdots \\
u_r & u_r^\sigma & \cdots & u_r^{\sigma^{r-1}}
\end{pmatrix} \right|^{-1}  \begin{pmatrix}
x & x^\sigma & \cdots & x^{\sigma^r} \\
u_1 & u_1^\sigma & \cdots & u_1^{\sigma^r} \\
\vdots \\
u_r & u_r^\sigma & \cdots & u_r^{\sigma^r}
\end{pmatrix},\]}
is a $\sigma$-subspace polynomial.
Clearly, two different $\sigma$-subspace polynomials have distinct kernels.
\end{proof}

\begin{remark}\label{new}
When $\F_{q^n} \subseteq \F_{q^m}$ and $\sigma$ is a generator of $\mathrm{Gal}(\F_{q^m}:\F_q)$, by \cite[Theorem 5]{GQ2009} a $\sigma$-subspace polynomial in $\mathcal{L}_{n,\sigma}$ is also a $\sigma$-subspace polynomial in $\mathcal{L}_{m,\sigma}$.
\end{remark}

A bound on family of $\sigma$-subspace polynomials agreeing on the last coefficients has been provided.

\begin{lemma}\cite[Lemma 5]{TrombZullo}\label{lemma:subspaceset}
Let $g,r,n$ and $m \in \mathbb{Z}^+$ be positive integers such that $g \leq r < n \leq m$.
Let $S$ be a subset of $\F_{q^m}$ of $n$ $\F_q$-linearly independent elements and let denote by $\tilde{\mathcal{P}}_{r,\sigma}$ the subset of $\mathcal{P}_{r,\sigma}$ whose polynomials have kernel contained in the $\fq$-subspace spanned by $S$.
There exists a subset $\mathcal{F} \subset \tilde{\mathcal{P}}_{r,\sigma}$ of $\sigma$-subspace polynomials coinciding on the last $g$ $\sigma$-coefficients, such that
\[ |\mathcal{F} | \geq \frac{{n \brack r}_q}{q^{m(g-1)}}. \]
\end{lemma}

The family of $\sigma$-subspace polynomials is closed under the adjoint operation (in the sense of Proposition \ref{prop:adjointnew}) and the composition (up to a scalar multiple) with the maps of shape $\tau_{\alpha}\colon x \in \F_{q^m}\mapsto \alpha x \in \F_{q^m}$, with $\alpha\in \F_{q^m}$.

\begin{proposition}\label{prop:adjointnew}
Let $f(x)=\sum_{i=0}^{k}a_i x^{\sigma^i} \in \mathcal{L}_{m,\sigma}$ be a $\sigma$-subspace polynomial. The $\sigma$-linearized polynomial $g(x)=(\hat{f}(x))^{\sigma^k}=\sum_{i=0}^k (a_{k-i}x)^{\sigma^i}$ is a $\sigma$-subspace polynomial.
\end{proposition}
\begin{proof}
As, $\dim_{\fq}(\ker(f))=k$, the assertion follows by Remark \ref{rem:adjoint}.
\end{proof}

\begin{proposition} \label{polalphav}
Let $f(x) \in \mathcal{L}_{m,\sigma}$ be a $\sigma$-subspace polynomial having $\sigma$-degree $k$. For every $\alpha \in \F_{q^m}^*$, $f_\alpha(x)=\alpha^{\sigma^k}f(\alpha^{-1}x)$ is a $\sigma$-subspace polynomial.
In particular, the monomials that appear in $f_{\alpha}(x)$ with non-zero coefficient are the same as $f(x)$.
\end{proposition}

\begin{proof}
An element $y \in \F_{q^m}$ is in $\ker(f_{\alpha})$ if and only if $f(\alpha^{-1}y)=0$, that happens if and only if $y \in \alpha \ker(f)$.
Beause of $\deg_{\sigma}(f_{\alpha}(x))=\deg_{\sigma}(f(x))$, we get that $f_{\alpha}(x)$ is a $\sigma$-subspace polynomial.
\end{proof}

\subsection{Known examples of $\sigma$-subspace polynomials}\label{Sec:subpoly}

As seen before, it is quite easy to construct examples of $\sigma$-subspace polynomials, whereas it is hard (and more useful) to construct them by using relations on their coefficients.
By looking to the latter point of view, very few examples are known. This subsection is devoted to the description of the actually known families of subspace polynomials. From now we suppose that $\sigma$ is also a generator of $\mathrm{Gal}(\F_{q^m}:\F_q)$.

\begin{proposition} \label{prop:binom}
Let $t$ be a positive integer such that $1 \leq t \leq n-1$ and $t\mid n$, let $s$ be a positive integer coprime with $n$ and let $\sigma\colon x \in \fqn \mapsto x^{q^s}\in \fqn$.
The set
\[
\mathcal{N}=\{x^{\sigma^t}-a_0x \colon a_0 \in \F_{q^n} \mbox{and } \N_{q^n/q^t}(a_0)=1\} \subseteq \mathcal{L}_{n,\sigma}
\]
is a family of $\sigma$-subspace polynomials of $\sigma$-degree $t$ and
\[ |\mathcal{N}| = \frac{q^{n}-1}{q^t-1}.\] 
In particular, if $m$ is a multiple of $n$ and $\gcd(s,m)=1$, then $\mathcal{N}$ is a set of $\frac{q^n-1}{q^t-1}$ $\sigma$-subspace polynomials of $\mathcal{L}_{m,\sigma}$ whose elements have $\sigma$-degree $t$.
\end{proposition}

\begin{proof}
It is well-known that any binomials $x^{\sigma^t}-a_0x \in \mathcal{L}_{n,\sigma}$ with $\N_{q^n/q^t}(a_0)=1$ is a $\sigma$-subspace polynomial, see e.g.\ \cite[Corollary 3.5]{teoremone}. 
The second part follows from Remark \ref{new}.
\end{proof}



Similarly, in \cite{raviv_2016} and in \cite{TrombZullo} the following family of $\sigma$-subspace polynomials was introduced.

\begin{proposition} (\cite[Construction 2]{raviv_2016},\cite[Proposition III.3]{TrombZullo}) \label{lemma:trace}
Let $t$ be a positive integer such that $1 \leq t \leq n-1$ and $t\mid n$, let $s$ be a positive integer coprime with $n$ and let $\sigma\colon x \in \fqn \mapsto x^{q^s}\in \fqn$.
The set
\[ \mathcal{T}=\left\{ \sum_{i=0}^{\frac{n}t-1} \beta^{\sigma^{it}-\sigma^{n-t}} x^{\sigma^{it}} \colon \beta \in \F_{q^n}^* \right\}\subset \mathcal{L}_{n,\sigma} \]
is a set of $\sigma$-subspace polynomials whose elements have $\sigma$-degree $n-t$ and
\[ |\mathcal{T}| = \frac{q^{n}-1}{q^t-1}. \]
In particular, if $m$ is a multiple of $n$ and $\gcd(s,m)=1$, then $\mathcal{T}$ can be also seen as a set of $\frac{q^n-1}{q^t-1}$ $\sigma$-subspace polynomials of $\mathcal{L}_{m,\sigma}$ whose elements have $\sigma$-degree $n-t$.
\end{proposition}

Generalizing the results in \cite{McGuireMueller}, in \cite{SantZullo1} was presented a new class of $\sigma$-subspace polynomials.

\begin{theorem}(\cite[Theorem 1.1]{McGuireMueller},\cite[Theorem 1.3]{SantZullo1})\label{th:trin1}
Let $n=(t-1)t+1$ and $f(x)=x^{\sigma^t}-bx^\sigma-ax \in \mathcal{L}_{n,\sigma}$. If
\begin{itemize}
  \item $\N_{q^n/q}(a)=(-1)^{t-1}$;
  \item $b=-a^{\frac{\sigma^n-\sigma}{\sigma^t-1}}$;
  \item $t-1$ is a power of the characteristic of $\F_{q^n}$,
\end{itemize}
then $f(x)$ is a $\sigma$-subspace polynomial.
\end{theorem}

Therefore, we can derive the following set of $\sigma$-subspace polynomials.

\begin{corollary}\label{coro:trin}
Let $t$ and $n \in \mathbb{Z}^+$ be positive integers such that $n=(t-1)t+1$ and $t-1$ is a power of the characteristic of $\F_{q^n}$.
Let $s$ be a positive integer coprime with $n$ and let $\sigma\colon x \in \fqn \mapsto x^{q^s}\in \fqn$.
The set
\[ \mathrm{Tri}_1=\left\{ x^{\sigma^t}-bx^\sigma-ax \colon a,b \in \F_{q^n}, \,\, \N_{q^n/q}(a)=(-1)^{t-1} \right.\]
\[\left. \text{and}\,\, b=-a^{\frac{\sigma^n-\sigma}{\sigma^t-1}} \right\} \subset {\mathcal L}_{n,q}  \]
is a set of $\frac{q^n-1}{q-1}$ $\sigma$-subspace polynomials whose elements have $\sigma$-degree $t$. In particular, if $m$ is a multiple of $n$ and $\gcd(s,m)=1$, then $\mathrm{Tri}_1$ is a set of $\frac{q^n-1}{q-1}$ $\sigma$-subspace polynomials of $\mathcal{L}_{m,\sigma}$ whose elements have $\sigma$-degree $t$.
\end{corollary}
\begin{proof}
By Theorem \ref{th:trin1}, it follows that the polynomials in $\mathrm{Tri}_1$ are $\sigma$-subspace polynomials in $\mathcal{L}_{n,\sigma}$ and the cardinality of $\mathrm{Tri}_1$ exactly coincides with the number of elements of $\F_{q^n}$ with norm over $\fq$ equals to $(-1)^{t-1}$, that is $\frac{q^n-1}{q-1}$.  
The second part follows from Remark \ref{new}.
\end{proof}

In the same hypothesis of Corollary \ref{coro:trin}, by Remark \ref{rem:adjoint}, we obtain a further family of $\sigma$-subspace trinomials.

\begin{corollary} \label{cor:trinadj}
Let $t$ and $n \in \mathbb{Z}^+$ be positive integers such that $n=(t-1)t+1$ and $t-1$ is a power of the characteristic of $\F_{q^n}$.
Let $s$ be a positive integer coprime with $n$ and let $\sigma\colon x \in \fqn \mapsto x^{q^s}\in \fqn$.
The set
\[ \widehat{\mathrm{Tri}_1}=\left\{ x+(-b)^{\sigma^{t-1}}x^{\sigma^{t-1}}+(-a)^{\sigma^t}x^{\sigma^t} \colon a,b \in \F_{q^n}, \,\, \N_{q^n/q}(a)=(-1)^{t-1} \right.\]
\[\left. \text{and}\,\, b=-a^{\frac{\sigma^n-\sigma}{\sigma^t-1}} \right\} \subset {\mathcal L}_{n,\sigma} \]
is a set of $\frac{q^n-1}{q-1}$ $\sigma$-subspace polynomials whose elements have $\sigma$-degree $t$. In particular, if $m$ is a multiple of $n$ and $\gcd(s,m)=1$, then $\widehat{\mathrm{Tri}_1}$ is a set of $\frac{q^n-1}{q-1}$ $\sigma$-subspace polynomials of $\mathcal{L}_{m,\sigma}$ whose elements have $\sigma$-degree $t$.
\end{corollary}

A further family of $\sigma$-subspace polynomials has been found in \cite{SantZullo1}.

\begin{theorem}\cite[Theorem 3.1]{SantZullo1}\label{th:trin2}
Let $n=t^2-1$ and $f(x)=x^{\sigma^t}-bx^\sigma-ax \in \mathcal{L}_{n,q}$, with $q$ power of $2$. If
\begin{itemize}
  \item $\N_{q^n/q}(a)=1$;
  \item $b=a^{-\frac{\sigma^{t^2}-\sigma^t}{\sigma^t-1}}$;
  \item $t$ is a power of $2$,
\end{itemize}
then $f(x)$ is a $\sigma$-subspace polynomial of $\mathcal{L}_{n,\sigma}$.
\end{theorem}

As a consequence, we get the following two sets of $\sigma$-subspace polynomials, obtained also using Remark \ref{rem:adjoint}.

\begin{corollary}\label{coro:trin2}
Let $t$ and $n \in \mathbb{Z}^+$ be positive integers such that $n=t^2-1$ and let $t$ and $q$ be both a power of $2$.
Let $s$ be a positive integer coprime with $n$ and let $\sigma\colon x \in \fqn \mapsto x^{q^s}\in \fqn$.
The set
\[ \mathrm{Tri}_2=\left\{ x^{\sigma^t}-bx^\sigma-ax \colon a,b \in \F_{q^n}, \,\, \N_{q^n/q}(a)=1 \right.\]
\[\left. \text{and}\,\, b=a^{-\frac{\sigma^{t^2}-\sigma^t}{\sigma^t-1}} \right\} \subset {\mathcal L}_{n,q} \]
and the set
\[ \widehat{\mathrm{Tri}_2}=\left\{ x+(-b)^{\sigma^{t-1}}x^{\sigma^{t-1}}+(-a)^{\sigma^t}x^{\sigma^t} \colon a,b \in \F_{q^n}, \,\, \N_{q^n/q}(a)=1 \right.\]
\[\left. \text{and}\,\, b=a^{-\frac{\sigma^{t^2}-\sigma^t}{\sigma^t-1}} \right\} \subset {\mathcal L}_{n,q}  \]
are sets of $\frac{q^n-1}{q-1}$ $\sigma$-subspace polynomials whose elements have $\sigma$-degree $t$. In particular, if $m$ is a multiple of $n$ and $\gcd(s,m)=1$, then $\mathrm{Tri}_2$ and $\widehat{\mathrm{Tri}_2}$ can be also seen as sets of $\frac{q^n-1}{q-1}$ $\sigma$-subspace polynomials of $\mathcal{L}_{m,\sigma}$ whose elements have $\sigma$-degree $t$.
\end{corollary}



\subsection{Generalizing the trace function}\label{sec:gentrace}

Huang et al. in \cite{HKP}, introduce the following family of $\sigma$-subspace polynomials. 

\begin{theorem}\cite[Lemma 4.3]{HKP} \label{teo:subspacehuang} Let $p$ be a prime, $q$ be a power of $p$ and $k$ be a positive integer. Consider $q'=p^r$, where $r$ is a non-negative integer, and let $p_i=1+q'+q'^2\ldots+q'^i$, with $i \in \{0,\ldots,k\}$. Then the $q$-polynomial
$$f(x)=x+\sum_{i=0}^{k-1} x^{q^{p_i}}$$
is a subspace polynomial of $\mathcal{L}_{n,q}$, with $n=p_k$.
\end{theorem}

Note that those polynomials generalize the trace function, in the sense that when $r=0$ it follows that $f(x)=\mathrm{Tr}_{q^n/q}(x)$. In the following proposition we will make use of \cite[Theorem 3.62]{lidl_finite_1997} which states that for any $f(x)=\sum_{i=0}^k a_ix^{q^i}, g(x)=\sum_{i=0}^h a_ix^{q^i} \in \mathcal{L}_{n,q}$ then 
\[ f(x) \mid g(x) \Leftrightarrow \sum_{i=0}^k a_ix^i \mid \sum_{j=0}^h b_j x^j. \]

\begin{proposition}\label{prop:extensionnew}
Let $p$ be a prime, $q$ be a power of $p$ and $k,t$ be positive integers. Consider $q'=p^r$, where $r$ is a non-negative integer, and let $p_i=1+q'+q'^2\ldots+q'^i$, with $i \in \{0,\ldots,k\}$. Then the $q$-polynomial
$$f(x)=x+\sum_{i=0}^{k-1} x^{q^{tp_i}}$$
is a subspace polynomial of $\mathcal{L}_{n,q}$, with $n=tp_k$.
\end{proposition}

\begin{proof}
By Theorem \ref{teo:subspacehuang}, the $q$-polynomial 
$$f(x)=x+\sum_{i=0}^{k-1} x^{q^{p_i}}$$
is a subspace polynomial of $\mathcal{L}_{h,q}$, with $h=p_{k}$. Then $f(x) \mid x^{q^h}-x$. By \cite[Theorem 3.62]{lidl_finite_1997},
$1+\sum_{i=0}^{k-1} x^{{p_i}} \mid x^h-1$ if and only if $1+\sum_{i=0}^{k-1} x^{t{p_i}} \mid x^{th}-1$, which turns out to be equivalent to  $$x+\sum_{i=0}^{k-1} x^{q^{t{p_i}}} \mid x^{q^{th}}-x.$$ Then $f(x)=x+\sum_{i=0}^{k-1} x^{q^{t{p_i}}}$ is a subspace polynomial of $\mathcal{L}_{n,q}$, with $n=tp_k$.
\end{proof}

We note that when $r=0$, then $f(x)=\mathrm{Tr}_{q^n/q^t}(x)$.

\begin{proposition} \label{prop:subspacehuang}
Let $p$ be a prime, $q$ be a power of $p$ and $k,t$ be positive integers. Consider $q'=p^r$, where $r$ is a non-negative integer, let $p_i=1+q'+q'^2\ldots+q'^i$, with $i \geq 0$, and let $n=tp_k$. 
Let $s$ be a positive integer coprime with $n$ and let $\sigma\colon x \in \fqn \mapsto x^{q^s}\in \fqn$.
Then the $\sigma$-polynomial
$$f(x)=x+\sum_{i=0}^{k-1} x^{\sigma^{tp_i}}$$
is a $\sigma$-subspace polynomial of $\mathcal{L}_{n,\sigma}$.
\end{proposition}

\begin{proof}
Let $\overline{q}=q^s$. Clearly, $f(x)$ can be seen as a $\overline{q}$-polynomial of $\mathcal{L}_{n,\overline{q}}$.
By Proposition \ref{prop:extensionnew} it follows that $\dim_{\F_{\overline{q}}}(\ker(f(x)))= tp_{k-1}$, that is $f(x)$ has maximum kernel as element of $\mathcal{L}_{n,\overline{q}}$.
So that, by \cite[Theorem 1.2]{teoremone} and \cite[Theorem 7]{McGuireSheekey} we have
  \[ A_n=C_f\cdot C_f^{\overline{q}}\cdot\ldots\cdot C_f^{\overline{q}^{n-1}}=I_{tp_{k-1}}, \]
  where $C_f$ is the companion matrix associated to $f(x)$.
  As $C_f\in \F_{q^n}^{tp_{k-1}\times tp_{k-1}}$, $A_n$ coincides with $C_f \cdot C_f^{\sigma} \cdot \ldots \cdot C_f^{\sigma^{n-1}}$ and from \cite[Theorem 1.2]{teoremone} and \cite[Theorem 7]{McGuireSheekey} it follows that the $\sigma$-polynomial $f(x)$ has maximum kernel, that is $f(x)$ is a $\sigma$-subspace polynomial in $\mathcal{L}_{n,\sigma}$.
\end{proof}

In the next proposition we consider the adjoint of polynomials of the previous result, giving new $\sigma$-subspace polynomials.

\begin{proposition} \label{cor:adjoinsubfield}
Let $p$ be a prime, $q$ be a power of $p$ and $k,t$ be positive integers. Consider $q'=p^r$, where $r$ is a non-negative integer, let $p_i=1+q'+q'^2\ldots+q'^i$, with $i \in \{0,\ldots,k\}$, and let $n=tp_k$. 
Let $s$ be a positive integer coprime with $n$ and let $\sigma\colon x \in \fqn \mapsto x^{q^s}\in \fqn$.
Then the $\sigma$-polynomial
$$g(x)=x^{\sigma^{tp_{k-1}}}+\sum_{i=1}^{k-1} x^{\sigma^{t(p_{k-1}-p_i)}}$$
is a $\sigma$-subspace polynomial of $\mathcal{L}_{n,\sigma}$.
\end{proposition}

\begin{proof}
By applying Remark \ref{rem:adjoint} to the $\sigma$-subspace polynomials of Proposition \ref{prop:subspacehuang}, straightforward computation show the assertion.
\end{proof}

Now, we detect a new family of $\sigma$-subspace polynomials having the same nonzero coefficients.

\begin{theorem} \label{teo:familyquasisubfield}
Let $p$ be a prime, $q$ be a power of $p$ and $k,t$ be positive integers. Consider $q'=p^r$, where $r$ is a non-negative integer, let $p_i=1+q'+q'^2\ldots+q'^i$, with $i \in \{0,\ldots,k\}$, and let $n=tp_k$. 
Let $s$ be a positive integer coprime with $n$ and let $\sigma\colon x \in \fqn \mapsto x^{q^s}\in \fqn$.
Let $f(x)$ be as in Proposition \ref{prop:subspacehuang}, and $g(x)$ as in Proposition \ref{cor:adjoinsubfield}. Then the sets
\begin{equation} \label{eq:familyhuang}
\mathcal{Q}=\{ \beta^{-\sigma^{tp_{k-1}}} f(\beta x): \beta \in \F_{q^n}^* \},
\end{equation}
and 
\begin{equation} \label{eq:familyadjointsubfield}
    \mathcal{Q}'=\{ \beta^{-\sigma^{tp_{k-1}}} g(\beta x): \beta \in \F_{q^n}^* \}
\end{equation}
are families of $\sigma$-subspace polynomials of $\mathcal{L}_{n,\sigma}$ having size $\frac{q^n-1}{q^t-1}$.
\end{theorem}

\begin{proof}
By Proposition \ref{polalphav}, $$\beta^{-\sigma^{tp_{k-1}}}f(\beta x)=\beta^{-\sigma^{tp_{k-1}}+1}x+\sum_{i=0}^{k-1} \beta^{-\sigma^{tp_{k-1}}+\sigma^{tp_i}}x^{\sigma^{tp_i}}$$ 
is a $\sigma$-subspace polynomial of $\mathcal{L}_{n,\sigma}$. Now, suppose that
\begin{equation} \label{eq:equalityfamilyquasisub}
\alpha^{-\sigma^{tp_{k-1}}} f(\alpha x)=\beta^{-\sigma^{tp_{k-1}}} f(\beta x).
\end{equation}
Then comparing the coefficients of $\sigma$-degree $0$ implies $\alpha^{-\sigma^{tp_{k-1}}+1}=\beta^{-\sigma^{tp_{k-1}}+1}$. Raising to $\sigma^{tq'^k}$-th power, we obtain $\left( \frac{\alpha}{\beta} \right)^{\sigma^{tq'^k}-1}=1$, which implies $$\frac{\alpha}{\beta} \in \F_{q^{tq'^k}} \cap \F_{q^n}.$$ Then $\frac{\alpha}{\beta} \in \F_{q^t}$, since $\gcd(tq'k,tp_k)=t\gcd(q'^k,p_{k-1})=t$. Clearly if $\alpha = \gamma \beta$, with $\gamma \in \F_{q^t}$, then \eqref{eq:equalityfamilyquasisub} holds. Then $\alpha$ and $\beta$ define the same $\sigma$-subspace polynomial if and only if $\alpha/\beta \in \F_{q^t}$, so that $\vert \mathcal{Q}\vert = \frac{q^{n}-1}{q^t-1}$.
Similar arguments can be applied to the family $\mathcal{Q}'$.
\end{proof}

\section{Bounds on the list size of rank-metric codes containing Gabidulin codes}\label{sec:gen}

Our aim is to investigate rank-metric codes that contains generalized Gabidulin codes. First we recall the following result from \cite{wachter-zhe_2013}. 

\begin{theorem}\cite[Theorem 1]{wachter-zhe_2013}\label{thm:listdecGab}
Let $k, n$ and $m \in \mathbb{Z}^+$ such that $k \leq n \leq m$. Let $\mathrm{G}_{n,k,\sigma}$ be a generalized Gabidulin code with minimum distance $d=n-k+1$ obtained by evaluating $\mathcal{G}_{m,k,\sigma}\subseteq \mathcal{L}_{m,\sigma}$ over $n$ $\F_q$-linearly independent elements in $\F_{q^m}$. Let $\tau$ be an integer such that $\left\lfloor\frac{d-1}{2}\right\rfloor+1 \leq \tau \leq d-1$. Then, there exists a word $\mathbf{w} \in \F_{q^m}^n \setminus \mathrm{G}_{n,k,\sigma}$ such that
\[|\mathrm{G}_{n,k,\sigma} \cap B_{\tau}(\mathbf{w})| \geq \frac{{n \brack n-\tau}_q}{q^{m(n-\tau-k)}}.\]
\end{theorem}

A first (natural) bound on the list size of rank-metric codes containing generalized Gabidulin codes arises from the previous result.

\begin{theorem}\label{thm:listdecGab}
Let $h,k, n$ and $m \in \mathbb{Z}^+$ such that $h \leq k \leq n \leq m$. Let $\mathcal{C}$ be a rank-metric code of $\mathcal{L}_{m,\sigma}$ and let $\mathrm{C}$ be the evaluation code over $n$ arbitrary $\F_q$-linearly independent elements $\alpha_1,\dots,\alpha_n\in \F_{q^m}$ having minimum distance $d$. Suppose that $\mathcal{C}$ contains $(\mathcal{G}_{m,h,\sigma})^{\sigma^j}$, for some $j\leq m-h$ \emph{(\footnote{Here by $(\mathcal{G}_{m,k,\sigma})^{\sigma^j}$ we mean the set $\{f(x)^{\sigma^j} \colon f(x) \in \mathcal{G}_{m,k,\sigma} \}$.})}. Let $\tau$ be an integer such that $\left\lfloor\frac{d-1}{2}\right\rfloor+1 \leq \tau \leq d-1$. Suppose that $j < \tau$. Then, there exists a word $\mathbf{w} \in \F_{q^m}^n \setminus \mathrm{C}$ such that
\[|\mathrm{C} \cap B_{\tau}(\mathbf{w})| \geq \frac{{n \brack n-\tau}_q}{q^{m(n-\tau-h)}}.\]
\end{theorem}

\begin{proof}
Let $\mathrm{G}$ be the evaluation code of $(\mathcal{G}_{m,k,\sigma})^{\sigma^j}$ over $\alpha_1,\ldots,\alpha_n$.
Let $\tilde{\mathcal{P}}_{n-\tau,\sigma} \subset \mathcal{L}_{m,\sigma}$ be the set of $\sigma$-subspace polynomials of $\sigma$-degree $n-\tau$ whose kernels are $(n-\tau)$-dimensional $\fq$-subspaces of $\F_{q^m}$ contained in the $\fq$-subspace $\langle \alpha_1,\dots,\alpha_n \rangle_{\F_q}$.
By Lemma \ref{lemma:subspaceset}, there exists a subset ${\cal F}$ of $\tilde{\mathcal{P}}_{n-\tau,\sigma}$ whose elements coincide on the last $n-h-\tau+1$ coefficients, with size at least
\[ \frac{{n \brack n-\tau}_q}{q^{m(n-h-\tau)}}. \]
Precisely,
\begin{small}
\[{\cal F}=\left\{ \sum _{i=0}^{h-1} a_i x^{\sigma^i}+b_{h}x^{\sigma^{h}}+\cdots+b_{n-\tau-1}x^{\sigma^{n-\tau-1}}+x^{\sigma^{n-\tau}} \colon \right.\]
\[\left. (a_0,a_1,\ldots,a_{h-1}) \in \mathcal{A}  \right\}, \]
\end{small}
where $\mathcal{A}$ is a subset of $\F_{q^m}^{h}$ such that $|\mathcal{A}| \geq \frac{{n \brack n-\tau}_q}{q^{m(n-h-\tau)}}$, and the $b_j$ are fixed elements of $\F_{q^m}$.

Let $\mathcal{F}'=\{ f \circ x^{\sigma^j} \colon f \in \mathcal{F}\}$ and note that the $\sigma$-polynomials of ${\cal F}'$ are not $\sigma$-subspace polynomials, but they still have $q^{n-\tau}$ roots over $\F_{q^m}$ and are of form
\[ a_0x^{\sigma^j}+a_1x^{\sigma^{j+1}}+\ldots+a_{h-1}x^{\sigma^{h-1+j}}+b_{h}x^{\sigma^{h+j}}+\ldots+b_{n-\tau-1}x^{\sigma^{n-\tau+j-1}}+x^{\sigma^{n-\tau+j}}. \]

Let $R$ be a polynomial in $\mathcal{F}'$. First, we observe that $c_R \notin \mathrm{C}$, since $$\mathrm{rk}(c_{R})=n-\dim_{\F_q}(\ker(R)) = \tau < d.$$
Moreover, we have that $c_{R-P} \in \mathrm{C}$, for each $P \in \mathcal{F}'$, indeed $$j \leq \deg_{\sigma}({R-P}) \leq h-1+j$$ and so $c_{R-P} \in \mathrm{G}\subseteq \mathrm{C}$.

Now we observe that $c_{R-P} \in B_{\tau}(c_R)$ for every $P \in \mathcal{F}'$ because $\mathrm{rk}(c_R-c_{R-P})=\mathrm{rk}(c_P)=\tau$. 
Finally, if $P\neq P'$, with $P,P' \in \mathcal{F}'$, we have that $c_{R-P}\neq c_{R-P'}$.
So, choosing $\mathbf{w}=c_R$,
we have that
\[
|C \cap  B_{\tau}(\mathbf{w})| \geq |\mathcal{F}'|=|\mathcal{F}| \geq  \frac{{n \brack n-\tau}_q}{q^{m(n-h-\tau)}},
\]
and the assertion follows.
\end{proof}

As a particular case, if $C=\mathrm{H}_{n,k,\sigma}(f_1,f_2)$, \cite[Theorem 9]{TrombZullo} follows.

Under the hypothesis of Theorem \ref{thm:listdecGab}, it is straightforward to show that $\mathrm{C}$ cannot be list decoded efficiently at the radius $\tau$ if
\begin{equation}\label{eq:Johnson} 
\tau \geq \frac{m+n}{2}-\sqrt{\frac{(m+n)^2}{4} -m(n-h+1-\epsilon)},
\end{equation}
where $0\leq\epsilon<1$.

\begin{remark}
Note that the above result has been obtained by adapting the proof of \cite[Theorem 1]{wachter-zhe_2013}, but it is not an its direct consequence. 
Indeed, it is not known whether or not, once \cite[Theorem 1]{wachter-zhe_2013} is applied to the subcode, the word $\mathbf{w}$ is in the entire code containing the (generalized) Gabidulin code.
\end{remark}

In the following we give a bound on the list size of rank-metric codes containing generalized Gabidulin codes relying on the existence of certain families of $\sigma$-subspace polynomials. As a consequence we get negative results on the list decodability of these codes.

Let $l$ and $h$ be two positive integers such that $h\leq l \leq n-1 $.
Denote
\begin{equation} \label{eq:insiemepol}
 \mathrm{Pol}_{l,h}=\{x^{\sigma^l}+a_{h-1}x^{\sigma^{h-1}}+\cdots+a_1x^\sigma+a_0x:a_i \in \F_{q^n}\}\subseteq \mathcal{L}_{n,\sigma}.
\end{equation}

\begin{theorem} \label{teo:main}
Let $n, m \in \Z$ be positive integers such that
$n \mid m$. Let $\mathcal{C}$ be a rank-metric code of $\mathcal{L}_{m,\sigma}$ and let $\mathrm{C}$ be
the associated evaluation code over an $\fq$-basis of $\beta \F_{q^n}$, for some $\beta \in \F_{q^m}^*$, with minimum distance $d$. Let $l$ and $h$ be positive integers such that $n-d+1\leq l\leq n-\lfloor \frac{d-1}{2} \rfloor-1$ and $l \geq h$.
Suppose that
\begin{itemize}
    \item $(\mathcal{G}_{n,h,\sigma})^{\sigma^j} \subseteq \mathcal{C}$, with $j < n-l$.
    \item There exists a subset $Sub \subseteq \mathrm{Pol}_{l,h}$ of $\sigma$-subspace polynomials of size $g$.
\end{itemize}
Then there exists a word $\mathbf{w}\in \F_{q^m}^n \setminus \mathrm{C}$ such that 
$$
\lvert \mathrm{C} \cap B_{n-l}(\mathbf{w}) \rvert \geq g.
$$
\end{theorem}

\begin{proof}
For every $P \in Sub$, denote by $P_{\beta}$ the polynomial $\beta^{\sigma^l} P(\beta^{-1} x)$.
As $\ker(P)\subseteq \F_{q^n}$, $\ker(P_\beta)=\beta \ker(P)\subseteq \beta \F_{q^n}$. 
Let $Sub_{\beta}=\{P_{\beta} \colon P \in Sub\}$. It is easy to see that $\lvert Sub_{\beta} \rvert = \lvert Sub \rvert $.

Consider now, $\mathcal{S}=(Sub_{\beta})^{\sigma^j}=\{P_{\beta}(x)^{\sigma^j} \colon P \in Sub\}\subseteq  \mathrm{Pol}_{l,h}^{\sigma^j}$. Note that all the elements in $\mathcal{S}$ have kernel of dimension $l$, because of Proposition \ref{polalphav}.

Let $R \in \mathcal{S}$. Then $\ker(R) \subseteq \beta \F_{q^n}$, so that $$\mathrm{rk}(c_R)=n-\dim_{\fq}(\ker(R) \cap \beta \F_{q^n})$$ 
$$=n -\dim_{\fq}(\ker(R))=n-l<d,$$
and hence $c_R \notin \mathrm{C}$.

Moreover, for every $P \in \mathcal{S}$ we have
$$R-P \in (\mathcal{G}_{n,h,\sigma} )^{\sigma^j}\subseteq \mathcal{C}.$$ 

Furthermore, for every $P \in \mathcal{S}$, as $\ker(P) \subseteq \beta \F_{q^n}$, then $\mathrm{rk}(c_R-c_{R-P})=\mathrm{rk}(c_P)=n-\dim_{\fq}(\ker(P))=n-l$.
This implies that $c_{R-P} \in B_{n-l}(c_R)$, for each $P \in \mathcal{S}$ and hence
$$c_{R-P}\in \mathrm{G} \cap B_{n-l}(c_R) \subseteq \mathrm{C} \cap B_{n-l}(c_R),$$
where $\mathrm{G}$ is the evaluation code of $\mathcal{G}_{n,h,\sigma}^{\sigma^j}$.
Finally, we have to prove that different choices of $P \in \mathcal{S}$ lead to different codewords of $\mathrm{C}$. Let $P,P' \in \mathcal{S}$ with $P \neq P'$. If $c_{R-P}=c_{R-P'}$, then $c_{P-P'}=0$, that is the $\sigma$-polynomial $P(x)-P'(x)$ has at least $q^n$ roots but its $\sigma$-degree is at most $h+j-1<n$, a contradiction.
\end{proof}

As a consequence, if $Sub$ is big enough we may get negative results on the list decodability of such a kind of rank-metric codes.

\begin{corollary} \label{cor:notdecod}
Let $n, m \in \Z$ be positive integers such that
$n \mid m$. Let $\mathcal{C}$ be a rank-metric code of $\mathcal{L}_{m,\sigma}$ and let $\mathrm{C}$ be
the associated evaluation code over an $\F_q$-basis of $\beta \F_{q^n}$, for some $\beta \in \F_{q^m}^*$, with
minimum distance $d$. Let $l$ and $h$ be positive integers such that $n-d+1\leq l\leq n-\lfloor \frac{d-1}{2} \rfloor-1$ and $l \geq h$.
Suppose that
\begin{itemize}
    \item $(\mathcal{G}_{n,h,\sigma})^{\sigma^j} \subseteq \mathcal{C}$, for some $j \in \Z^+$ with $j < n-l$.
    \item There exists a subset $Sub \subseteq \mathrm{Pol}_{l,h}$ of $\sigma$-subspace polynomials of size $g=O(q^n)$.
\end{itemize}
Then $\mathrm{C}$ cannot be list decoded efficiently at any radius $\tau \geq n-l$.

Furthermore, if $l=n-\lfloor \frac{d-1}{2} \rfloor +1$, then $\mathrm{C}$ cannot be list decoded efficiently at all.
\end{corollary}

The same arguments apply for the following cases.

Let $l$ and $h$ be two positive integers such that $h\leq l \leq n-1 $.
Denote
\[ \widehat{\mathrm{Pol}}_{l,h}=\{a_{l}x^{\sigma^{l}}+\cdots+a_{l-h+1}x^{\sigma^{l-h+1}}+x:a_i \in \F_{q^n}\}\subseteq \mathcal{L}_{n,\sigma}.\]

\begin{theorem} \label{teo:main2}
Let $n, m \in \Z$ be positive integers such that
$n \lvert m$. Let $\mathcal{C}$ be a rank-metric code of $\mathcal{L}_{m,\sigma}$ and let $\mathrm{C}$ be
the associated evaluation code over an $\F_q$-basis of $\beta\F_{q^n}$, for some $\beta \in \F_{q^m}^*$, with
minimum distance $d$. Let $l$ and $h$ be  positive integers such that $n-d+1\leq l\leq n-\lfloor \frac{d-1}{2} \rfloor-1$.
Suppose that
\begin{itemize}
    \item $(\mathcal{G}_{n,h,\sigma})^{\sigma^{l-h+1+j}} \subseteq \mathcal{C}$, for some $j \in \Z^{+}$ with $j < n-2l+h-1$ and $l\geq h$.
    \item There exists a subset $\widehat{Sub} \subseteq \widehat{\mathrm{Pol}}_{l,h}$ of $\sigma$-subspace polynomials of size $g$.
\end{itemize}
Then there exists a word $\mathbf{w}\in \F_{q^m}^n \setminus \mathrm{C}$ such that 
$$
\lvert \mathrm{C} \cap B_{n-l}(\mathbf{w}) \rvert \geq g.
$$
\end{theorem}

\begin{corollary} \label{cor:notdecod2}
Let $n, m \in \Z$ be positive integers such that
$n \mid m$. Let $\mathcal{C}$ be a rank-metric code of $\mathcal{L}_{m,\sigma}$ and let $\mathrm{C}$ be
the associated evaluation code over an $\F_q$-basis of $\beta\F_{q^n}$, for some $\beta \in \F_{q^m}^*$, with minimum distance $d$. Let $l$ and $h$ be positive integers such that $n-d+1\leq l\leq n-\lfloor \frac{d-1}{2} \rfloor-1$.
Suppose that
\begin{itemize}
    \item $(\mathcal{G}_{n,h,\sigma})^{\sigma^{l-h+1+j}} \subseteq \mathcal{C}$, for some $j \in \Z^{+}$ with $j < n-2l+h-1$ and $l \geq h$. 
    \item There exists a subset $\widehat{Sub} \subseteq \widehat{\mathrm{Pol}}_{l,h}$ of $\sigma$-subspace polynomials of size $g=O(q^n)$.
\end{itemize}
Then $\mathrm{C}$ cannot be list decoded efficiently at any radius $\tau \geq n-l$.

Furthermore, if $l=n-\lfloor \frac{d-1}{2} \rfloor +1$, then $\mathrm{C}$ cannot be list decoded efficiently at all.
\end{corollary}

\section{Some consequences}\label{sec:appl1}

We now apply Theorems \ref{teo:main} together with Corollary \ref{cor:notdecod} to get bound on the list size and negative results on the list decodability of rank-metric codes containing some generalized Gabidulin codes. 

\begin{theorem}
 \label{teo:applbinom}
Let $n, m \in \Z$ be positive integers such that
$n \mid m$. Let $\mathcal{C}$ be a rank-metric code of $\mathcal{L}_{m,\sigma}$. Let $\mathrm{C}$ be the evaluation code of $\C$ over an $\F_q$-basis of $\beta\F_{q^n}$, for some $\beta \in \F_{q^m}^*$, and minimum distance $d$. Suppose that there exists a positive integer $t\geq 1$ such that:
\begin{itemize}
    \item[1.] $n-d+1\leq t\leq n-\lfloor \frac{d-1}{2} \rfloor-1$;
    \item[2.] $t \mid n$;
    \item[3.] $(\mathcal{G}_{n,1,\sigma})^{\sigma^j} \subseteq \mathcal{C}$, with $j< n-t$.
\end{itemize}
Then 
\begin{itemize}
    \item $\mathrm{C}$ cannot be list decoded efficiently at any radius $\tau \geq n-t$.
    \item If $n$ is even and $d=n-1$, then $\mathrm{C}$ cannot be list decoded efficiently at all.
\end{itemize}    
\end{theorem}

\begin{proof}
Consider
$$\mathrm{Pol}_{t,2}=\{-a_0x+x^{\sigma^t}:a_0 \in \F_{q^n}\}.$$
Let $Sub$ be the family $\mathcal{N}$ of $\sigma$-subspace polynomials of Proposition \ref{prop:binom} with $t \mid n$. By Theorem \ref{teo:main}, there exists a word $\mathbf{w}\in \F_{q^m}^n \setminus \mathrm{C}$ such that 
$$
\lvert \mathrm{C} \cap B_{n-t}(\mathbf{w}) \rvert \geq \frac{q^n-1}{q^t-1},
$$
which proves the first point.
The second point follows from the fact that $\frac{q^n-1}{q^t-1} \sim q^{n/2}$.

Now, let $n$ be even and $d=n-1$.
The unique decoding radius of $\mathrm{C}$ is 
$$\left\lfloor \frac{d-1}{2} \right\rfloor=\frac{n}{2}-1.$$
So, we may choose $t=\frac{n}{2}$, and by the first part we have that $C$ cannot be list decoded efficiently at any radius $\tau\geq\frac{n}{2}$.
\end{proof}

The proof Theorem \ref{teo:applbinom} may be modified using Theorem \ref{teo:main} applied to Proposition \ref{lemma:trace} and Corollaries \ref{coro:trin} and \ref{coro:trin2}, getting the following results.

\begin{theorem} \label{teo:usetrace}
Let $n, m \in \Z$ be positive integers such that
$n \mid m$. Let $\mathcal{C}$ be a rank-metric code of $\mathcal{L}_{m,\sigma}$. Let $\mathrm{C}$ be the evaluation code of $\mathcal{C}$ over an $\F_q$-basis of $\beta\F_{q^n}$, for some $\beta \in \F_{q^m}^*$, and minimum distance $d$. Suppose that there exists a positive integer $t \geq 1$ such that:
\begin{itemize}
    \item[1.] $\lfloor \frac{d-1}{2} \rfloor+1 \leq t \leq d-1$;
    \item[2.] $t \mid n$;
    \item[3.] $(\mathcal{G}_{n,n-2t+1,\sigma})^{\sigma^j} \subseteq \mathcal{C}$, with $j< t-1$. 
\end{itemize}
Then
\begin{itemize}
    \item The code $\mathrm{C}$ cannot be list decoded efficiently at any radius $\tau \geq t$.
    \item If $n=r\left(\lfloor \frac{d-1}{2} \rfloor+1 \right)$, with $r \in \mathbb{N}$, then $\mathrm{C}$ cannot be list decoded efficiently at all. 
    
    In particular if $d$ is even and $\frac{d}{2}\mid n$ then $\mathrm{C}$ cannot be list decoded efficiently at all.
\end{itemize}    
\end{theorem}

\begin{proof}
Choosing
$$\mathrm{Pol}_{n-t,n-2t+1}=\{-x^{\sigma^{n-t}}+a_{n-2t}x^{\sigma^{n-2t}}+\ldots+a_0x:a_i \in \F_{q^n}\}$$
and $Sub$ as the family of the Proposition  \ref{lemma:trace}, Theorem \ref{teo:main} implies the existence of a word $\mathbf{w}\in \F_{q^m}^n \setminus \mathrm{C}$ such that 
$$
\lvert \mathrm{C} \cap B_{t}(\mathbf{w}) \rvert \geq \frac{q^n-1}{q^t-1}.
$$
\end{proof}

The next results can be proven by using the same techniques of the above results.

\begin{theorem}\label{thm:RMcontGn,2}
Let $n, m \in \mathbb{Z}^+$ be positive integers such that $n \mid m$. Let $\mathcal{C}$ be a rank-metric code of $\mathcal{L}_{m,\sigma}$ and let $\mathrm{C}$ be its associated evaluation code over an $\F_q$-basis of $\beta\F_{q^n}$, for some $\beta \in \F_{q^m}^*$, with minimum distance $d$. Let $t \in \mathbb{Z}^+$ such that: 
\begin{itemize}  \item[1.] $n-d+1 \leq t \leq n-\left\lfloor\frac{d-1}{2}\right\rfloor-1$; 
\item[2.] $t-1$ is a power of the characteristic of $\F_{q^n}$;  
\item[3.] $n=t(t-1)+1$;
\item[4.] $\mathcal{C}$ contains $(\mathcal{G}_{n,2,\sigma})^{\sigma^j}$, with $j< n-t$.
\end{itemize}
Then
$\mathrm{C}$ cannot be list decoded efficiently at any radius $\tau \geq n-t$.
\end{theorem}



\begin{remark}\label{rem:ntau}
Once we fix the integer $n=t(t-1)+1$, for some $t$ in $\mathbb{Z}^+$, then $\displaystyle \tau\geq t^2-2t+1$.
\end{remark}

\begin{remark}\label{rm:size}
We observe that once you apply Theorem \ref{teo:main} using the family of $\sigma$-subspace polynomials of Corollary \ref{coro:trin}, if $\mathrm{C}=\mathrm{G}_{n,k,\sigma}$ with $k\geq 2$ and with constraints on the involved parameters as prescribed in Theorem \ref{thm:RMcontGn,2}, then there exists a word $\mathbf{w} \in \F_{q^m}^n \setminus \mathrm{C}$ such that
\[|\mathrm{C} \cap B_{n-t}(\mathbf{w})| \geq \frac{q^n-1}{q-1} \sim q^{n-1},\] which improves the list size provided in \cite[Theorem 3]{raviv_2016} and \cite[Theorem 8]{TrombZullo}, for any value of the radius $\tau$ greater than or equal to $t^2-2t+1$.
\end{remark}

\begin{theorem}\label{thm:RMcontGn,22}
Let $n, m \in \mathbb{Z}^+$ be positive integers such that $n \mid m$. Let $q$ a power of 2. Let $\mathcal{C}$ be a rank-metric code of $\mathcal{L}_{m,q}$ and let $\mathrm{C}$ be its associated evaluation code over an $\F_q$-basis of $\beta \F_{q^n}$, for some $\beta \in \F_{q^m}^*$, with minimum distance $d$. Let $t \in \mathbb{Z}^+$ such that: \begin{itemize}  
\item[1.] $n-d+1 \leq t \leq n-\left\lfloor\frac{d-1}{2}\right\rfloor-1$ 
\item[2.] $t$ is a power of $2$; 
\item[3.] $n=t^2-1$; 
\item[4.] $\mathcal{C}$ contains $(\mathcal{G}_{n,2,\sigma})^{\sigma^j}$, with $j< n-t$.\end{itemize}
Then 
$\mathrm{C}$ cannot be list decoded efficiently at any radius $\tau \geq n-t$. 

\end{theorem}




\begin{remark}\label{rm:size2}
As for Remark \ref{rm:size}, we observe that if $\mathrm{C}=\mathrm{G}_{n,k,\sigma}$ with $k\geq 2$ and with constraints on the involved parameters as prescribed in Theorem \ref{thm:RMcontGn,2}, then there exists a word $\mathbf{w} \in \F_{q^m}^n \setminus \mathrm{C}$ such that
\[|\mathrm{C} \cap B_{n-t}(\mathbf{w})| \geq \frac{q^n-1}{q-1} \sim q^{n-1},\] which improves the list size provided in \cite[Theorem 3]{raviv_2016} and \cite[Theorem 8]{TrombZullo}, for any value of the radius $\tau$ greater than or equal to $t^2-t-1$.
\end{remark}

Now we consider the families of subspace polynomials introduced in Subsection \ref{sec:gentrace}.

\begin{theorem} \label{teo:usesubfield}
Let $p$ be a prime, $q$ be a power of $p$ and $k,t$ be positive integers. Consider $q'=p^r$, where $r$ is a non-negative integer, let $p_i=1+q'+q'^2\ldots+q'^i$, with $i \geq 0$, let $n=tp_k$ and $m$ be a multiple of $n$. 
Let $s$ be a positive integer coprime with $n$ and let $\sigma\colon x \in \fqn \mapsto x^{q^s}\in \fqn$.
Let $\mathcal{C}$ be a rank-metric code of $\mathcal{L}_{m,\sigma}$ and let $\mathrm{C}$ be the evaluation code of $\mathcal{C}$ over an $\F_q$-basis of $\beta\F_{q^n}$, for some $\beta \in \F_{q^m}^*$, and minimum distance $d$. Suppose that:
\begin{itemize}
    \item[1.] $\lfloor \frac{d-1}{2} \rfloor+1 \leq tq'^k \leq d-1$;
    \item[2.] $(\mathcal{G}_{n,tp_{k-2}+1,\sigma})^{\sigma^j} \subseteq \mathcal{C}$, with $j< tq'^k$. 
\end{itemize}
Then
$\mathrm{C}$ cannot be list decoded efficiently at any radius $\tau \geq tq'^k$. 
\end{theorem}

\begin{proof}
Consider 
$$\mathrm{Pol}_{tp_{k-1},tp_{k-2}+1}=\{x^{\sigma^{tp_{k-1}}}+a_{tp_{k-2}}x^{\sigma^{tp_{k-2}}}+\cdots+a_0x:a_i \in \F_{q^n}\}\subseteq \mathcal{L}_{n,\sigma}.
$$ 

By applying Theorem \ref{teo:main} choosing  $Sub=\mathcal{Q}$, the family described in Theorem \ref{teo:familyquasisubfield}, there exists a word $\mathbf{w}\in \F_{q^m}^n \setminus \mathrm{C}$ such that 
$$
\lvert \mathrm{C} \cap B_{n-tp_{k-1}}(\mathbf{w}) \rvert \geq |\mathcal{Q}|= \frac{q^n-1}{q^t-1},
$$
which proves the first point.
The second point follows from the fact that $\frac{q^n-1}{q^t-1} \sim q^{n/2}$.
\end{proof}

The previous results extends \cite[Theorem 4]{raviv_2016}, in the sense that when $r=0$ and $\mathcal{C}$ is a Gabidulin code, we obtain \cite[Theorem 4]{raviv_2016}.
Now, applying Theorem \ref{teo:main} to the family $\mathcal{Q}'$ of Theorem \ref{teo:familyquasisubfield}, the following result holds.

\begin{theorem} \label{teo:usesubfield2}
Let $p$ be a prime, $q$ be a power of $p$ and $k,t$ be positive integers. Consider $q'=p^r$, where $r$ is a non-negative integer, let $p_i=1+q'+q'^2\ldots+q'^i$, with $i \geq 0$, let $n=tp_k$ and $m$ be a multiple of $n$. 
Let $s$ be a positive integer coprime with $n$ and let $\sigma\colon x \in \fqn \mapsto x^{q^s}\in \fqn$.
Let $\mathcal{C}$ be a rank-metric code of $\mathcal{L}_{m,\sigma}$. Let $\mathrm{C}$ be the evaluation code of $\mathcal{C}$ over an $\F_q$-basis of $\beta\F_{q^n}$, for some $\beta \in \F_{q^m}^*$, and minimum distance $d$. Suppose that:
\begin{itemize}
    \item[1.] $\lfloor \frac{d-1}{2} \rfloor+1 \leq tq'^k \leq d-1$;
    \item[2.] $(\mathcal{G}_{n,t(p_{k-1}-p_0)+1,\sigma})^{\sigma^j} \subseteq \mathcal{C}$, with $j< tq'^k$. 
\end{itemize}
Then
$\mathrm{C}$ cannot be list decoded efficiently at any radius $\tau \geq tq'^k$. 
\end{theorem}

\section{Explicit examples}\label{sec:appl2}

In this section we will see explicit examples of rank-metric codes which cannot be list decoded efficiently from a certain radius or at all in some cases.

We start by showing how some known results on generalized Gabidulin codes and further families of codes may be obtained as corollaries of our results.

\begin{theorem}(\cite[Theorem 4]{raviv_2016},\cite[Theorem 10 and Corollary 11]{TrombZullo})\label{thm:listdecGabatall}
Let $k, t, n$ and $m \in \mathbb{Z}^+$ positive integers such that $k \leq n$, $t \mid n $ and $n \mid m$. 
Let $\mathrm{G}_{n,k,\sigma}$ be the evaluation code of $\mathcal{G}_{m,k,\sigma}$ over an $\F_q$-basis of $\beta \F_{q^m}$, with $\beta \in \F_{q^m}^*$. 

\noindent If $\left\lfloor \frac{d-1}{2} \right\rfloor+1 \leq t \leq  d-1$, then
\begin{itemize}
\item The code $\mathrm{G}_{n,k,\sigma}$ cannot be list decoded efficiently at any radius $\tau \geq t$;
\item In particular, any generalized Gabidulin code $\mathrm{G}_{n,k,\sigma}$ with minimum distance $d=2t$, cannot be list decoded efficiently at all.
\end{itemize}
\end{theorem}
\begin{proof}
As $t\geq \left\lfloor \frac{d-1}2 \right\rfloor+1$, we have that $k\geq n-2t+1$, so that $\mathrm{G}_{n,n-2t+1,\sigma} \subseteq \mathrm{G}_{n,k,\sigma}$.
So, the assertion then follows by applying Theorem \ref{teo:usetrace}. 
\end{proof}

The above proof may be adapted to get a similar result for other examples of MRD-codes as proved in \cite{TrombZullo}, which we slightly generalize.

\begin{theorem}\label{thm:listdecGenatallref}
Let $k,t,n$ and $m \in \mathbb{Z}^+$ such that $k \leq n$, $t \mid n$ and $n \mid m$. Let $\mathcal{C}=\mathcal{H}_{n,k,\sigma}(f_1,f_2)$ as in \eqref{eq:codes} and let $\mathrm{C} = \mathrm{H}_{n,k,\sigma}(f_1,f_2)$ be the associated evaluation code over an $\F_q$-basis of $\beta \F_{q^n}$, for some $\beta \in \F_{q^m}^*$, where $f_2$ is not the zero polynomial.
If $d(\mathrm{C})=d=n-k+1$, $t \leq  d-1$ and 
\[ t\geq \left\{ \begin{array}{ll} \left\lfloor \frac{d-1}{2} \right\rfloor+1 & \text{if}\,\, n-k\,\,\text{is even},\\  \\
\left\lfloor \frac{d-1}{2} \right\rfloor+2 & \text{if}\,\, n-k\,\,\text{is odd,}\\ 
\end{array} \right. \]
then
\begin{itemize}
    \item The code $\mathrm{C}$ cannot be list decoded efficiently at any radius $\tau \geq t$.
    \item  If $d=2t-2$, then $\mathrm{C}$ cannot be list decoded efficiently at any radius $\tau$ greater than or equal to $\left\lfloor \frac{d-1}2 \right\rfloor+2$.
    \item  If $d=2t-1$, then $\mathrm{C}$ cannot be list decoded efficiently at all.
\end{itemize}
\end{theorem}
\begin{proof}
Because of the hypothesis on $t$, $\mathrm{G}_{n,n-2t+1,\sigma} \subseteq \mathrm{G}_{n,k-1,\sigma} \subset \mathrm{H}_{n,k,\sigma}(f_1,f_2)$ and then Theorem \ref{teo:usetrace} applies. 
\end{proof}

Theorem \ref{thm:listdecGenatallref} generalizes \cite[Theorem 12 and Corollary 13]{TrombZullo} as the codes considered in \cite{TrombZullo} are evaluated over an $\fq$-basis of $\F_{q^n}$ and the last point of Theorem \ref{thm:listdecGenatallref} guarantees the existence of parameters such that the evaluation codes associated to the family of MRD-codes of shape \eqref{eq:codes} cannot be list decoded efficiently at all.

When $\mathcal{C}=\mathcal{H}_{n,k,\sigma}(f_1,f_2)$ defined as in \eqref{eq:codes} is not an MRD-code, then the associated evaluation code over $n$ $\F_q$-linearly independent elements may have minimum distance equals to either $n-k+1$ or $n-k$. In the former case Theorem \ref{thm:listdecGenatallref} applies. 
The latter case is covered by the following result.

\begin{theorem}
Let $k,t,n$ and $m \in \mathbb{Z}^+$ such that $k \leq n$, $t \mid n$ and $n \mid m$. Let $\mathcal{C}=\mathcal{H}_{n,k,\sigma}(f_1,f_2)$ as in \eqref{eq:codes} and let $\mathrm{C} = \mathrm{H}_{n,k,\sigma}(f_1,f_2)$ be the associated evaluation code over an $\F_q$-basis of $\beta \F_{q^n}$, for some $\beta \in \F_{q^m}^*$.
If $d(\mathrm{C})=d=n-k$ and $\left\lfloor \frac{d-1}{2} \right\rfloor+1 < t \leq  d-1$, then
\begin{itemize}
    \item The code $\mathrm{C}$ cannot be list decoded efficiently at any radius $\tau \geq t$.
    \item  If either $d=2t-2$ or $d=2t-3$, then $\mathrm{C}$ cannot be list decoded efficiently at any radius $\tau \geq \left\lfloor \frac{d-1}{2} \right\rfloor+2$.
\end{itemize}
\end{theorem}

\begin{proof}
Again the result follows by applying Theorem \ref{teo:usetrace} to the code $\mathrm{C}$ since $\mathrm{G}_{n,n-2t+1,\sigma} \subseteq \mathrm{G}_{n,k-1,\sigma}\subset \mathrm{H}_{n,k,\sigma}(f_1,f_2)$.
\end{proof}

Also, Theorem \ref{teo:usetrace} applies to the following class of rank-metric codes, for which we may determine the minimum distance when $m$ is large enough. 
In particular, we will make use of the results in \cite{BZ}, which are written in terms of \emph{Moore exponent sets}. More precisely, we use they results since the set $\{i_0,\ldots,i_k\}\subseteq \mathbb{Z}/n\mathbb{Z}$ is a Moore exponent set if and only if $\langle x^{\sigma^{i_0}},\ldots,x^{\sigma^{i_k}}\rangle_{\fqn}$ is an MRD-code in $\mathcal{L}_{n,\sigma}$.

\begin{lemma} \label{lemma.C_j}
Let $k$ and $m \in \mathbb{Z}^+$ such that $k\leq m-1$ and let $\sigma\colon x \in \mathbb{F}_{q^m}\mapsto x^{q^s}\in \F_{q^m}$ be a generator of $\mathrm{Gal}(\mathbb{F}_{q^m}\colon \fq)$.
Let
\[ \C_j=\langle x^{\sigma^i} \colon i \in\{0,\ldots,k-1\}\setminus\{j\}\rangle_{\F_{q^m}}\subseteq \mathcal{L}_{m,\sigma}, \]
for some $j \in \{1,\ldots,k-2\}$.
If $q>5$ and 
\begin{equation}\label{eq:m} m> 
\left\{ \begin{array}{ll} 12s+2, & \text{if}\,\, k=3,\\
\frac{13}3 sk+\log_q(13\cdot 2^{\frac{10}3}), & \text{if}\,\, k\geq 4,\\
\end{array}\right. \end{equation}
then the minimum distance of $\C_j$ is $m-k$.
In particular, $\C_j$ is not an MRD-code.
\end{lemma}
\begin{proof}
In \cite[Theorems 1.1, 3.2 and 4.1]{BZ}, the authors proved that if $m$ satisfies \eqref{eq:m}, then $\C_j$ is an MRD-code if and only if either $j=0$ or $j=k-1$, that is $\C_j$ is not an MRD-code. 
The assertion then follows noting that the elements in $\C_j$ have rank greater than or equal to $m-k$ and there should exist at least one element of rank $m-k$, otherwise $\C_j$ would be an MRD-code.
\end{proof}


Let $\mathrm{C}$ be the associated evaluation code of $\mathcal{C}_j$ over an $\F_q$-basis of $\beta \F_{q^n}$, with $\beta \in \F_{q^m}^*$. Suppose that $\mathcal{C}_j$ has minimum distance $d(\mathcal{C}_j)=m-k$, then, by Lemma \ref{lemma:rankmetricvect}, we have that $n-k \leq d(\mathrm{C}) \leq n-k+1$.

Now, as $\mathcal{C}_j$ contains $\mathcal{G}_{m,j,\sigma}$ and $\mathcal{G}_{m,k-j-1,\sigma}^{\sigma^{j+1}}$, the next results follows by Theorem \ref{teo:usetrace}.

\begin{theorem}
Let $k,t,n$ and $m \in \mathbb{Z}^+$ such that $k \leq n$, $t \mid n$ and $n \mid m$. Let $\mathcal{C}_j$ the code defined in Lemma \ref{lemma.C_j} and let $\mathrm{C}$ be the associated evaluation code over an $\F_q$-basis of $\beta \F_{q^n}$, for some $\beta \in \F_{q^m}^*$. Let $M=\max\{j,k-j-1\}$. Suppose that $t\geq \frac{n-j+1}{2}$, if $M=j$ or $t \geq \frac{n-k+j}{2}+1$ if $M=k-j-1$. \\
$\left\lfloor \frac{d-1}{2} \right\rfloor+1\leq t \leq  d-1$.
Then
the code $\mathrm{C}$ cannot be list decoded efficiently at any radius $\tau \geq t$.
\end{theorem}

In the case in which $j=k-2$, we can improve the previous result.

\begin{theorem}
Let $k,t,n$ and $m \in \mathbb{Z}^+$ such that $k \leq n$, $t \mid n$ and $n \mid m$. Let $\mathcal{C}_{k-2}$ be the code defined in Lemma \ref{lemma.C_j} and let $\mathrm{C}$ be the associated evaluation code over an $\F_q$-basis of $\beta \F_{q^n}$, for some $\beta \in \F_{q^m}^*$. \\
If $d(\mathrm{C})=d=n-k+1$, $t \leq  d-1$ and 
\[ t\geq \left\{ \begin{array}{ll} \left\lfloor \frac{d-1}{2} \right\rfloor+1 & \text{if}\,\, n-k\,\,\text{is even},\\  \\
\left\lfloor \frac{d-1}{2} \right\rfloor+2 & \text{if}\,\, n-k\,\,\text{is odd,}\\ 
\end{array} \right. \]
then
\begin{itemize}
    \item The code $C$ cannot be list decoded efficiently at any radius $\tau \geq t$.
    \item  If $d=2t-2$, then $\mathrm{C}$ cannot be list decoded efficiently at any radius $\tau$ greater than or equal to $\left\lfloor \frac{d-1}2 \right\rfloor+2$.
    \item  If $d=2t-1$, then $\mathrm{C}$ cannot be list decoded efficiently at all.
\end{itemize}

If $d(\mathrm{C})=d=n-k$ and $\left\lfloor \frac{d-1}{2} \right\rfloor+1 < t \leq  d-1$, then
\begin{itemize}
    \item The code $\mathrm{C}$ cannot be list decoded efficiently at any radius $\tau \geq t$.
    \item  If either $d=2t-2$ or $d=2t-3$, then $\mathrm{C}$ cannot be list decoded efficiently at any radius $\tau \geq \left\lfloor \frac{d-1}{2} \right\rfloor+2$.
\end{itemize}
\end{theorem}

\section{List decodability of costant dimension subspace codes}\label{sec:subspacecodes}

Consider the Grassmanian $\mathcal{G}_q(n,r)$, i.e. the set of all subspaces of dimension $r$ of $\F_{q^n}$.
A \emph{constant dimension subspace code} with parameters $(n,M_s,d_s,r)_q$ is a subset of $\mathcal{G}_q(n,r)$ with size $M_s$ and minimum subspace distance $d_s$ under the metric \[d_s(U,V)=\dim_{\F_q}(U)+\dim_{\F_q}(V)-2\dim_{\F_q}(U \cap V).\]
The interest in subspace codes has recently increased because of their application to error correction in random
network coding, see \cite{KoetterK}.
In \cite{silva} the authors produced a class of asymptotically optimal constant dimension subspace codes with the following procedure. 
Let $A \in \F_q^{n \times m}$ and denote by $\langle A \rangle$ the subspace spanned by the rows of a matrix $A$.

\begin{definition} \label{def:liftedcode}
Consider the mapping 
$$\mathcal{I}: \F_q^{n \times m} \rightarrow \mathcal{G}_q(n+m,n)$$
defined by
$$X \mapsto \langle [I_n \ X] \rangle,$$
where $I_n$ denotes the $n \times n$ identity matrix. The subspace $\mathcal{I}(X)=\langle [I_n \ X] \rangle $ is called lifting of the matrix $X$. 
\end{definition}

Hence, if $C$ is a rank-metric code in $\F_q^{n\times m}$, the \emph{lifting} of $C$ is
\[\mathcal{I}(C)=\{ \mathcal{I}(A) \colon A \in C \}.\] 

\begin{remark}
Note that if $\mathbf{w}\in \F_{q^m}^n$, we can see $\mathbf{w}$ as a matrix in $\F_q^{m\times n}$, once we fix an $\F_q$-basis of $\F_{q^m}$, so that we can define $\mathcal{I}(\mathbf{w})$.
If $\mathrm{C}$ is a rank-metric code in $\F_{q^m}^n$, we can look to $\mathrm{C}$ as a subset of matrices in $\F_q^{m\times n}$, so that we can define the lifting of $\mathrm{C}$.
\end{remark}

We will need the following property.


\begin{lemma} \cite[Proposition 4]{silva}
Let $C$ be a rank-metric code in $\F_q^{m\times n}$ with minimum distance $d_R$ and size $M_R$. 
Then\emph{(\footnote{$A^T$ denotes the transpose of $A$.})} 
$$\mathcal{I}(C^T)=\left\{ \mathcal{I}(A^T): A \in C \right\} $$
is a constant dimension subspace code with parameters $(n+m, M_s=M_R, d_s=2d_R,n)_q$.
\end{lemma}


Let $B_{\tau}^s( W )=\{ V : d_s( W, V ) \leq \tau \}$ denote the ball of radius $\tau$ centered at $W$ w.r.t.\ the subspace distance. Recall that \cite[Equation 3]{raviv_2016} establishes the following relation between the intersection of a rank-metric code $C$ and a fixed ball and its lifted subspace code $\mathcal{I}(C^T)$ and the correspondent ball:
\begin{equation}\label{eq:liftball}
 \lvert C \cap B_\tau (c_R)\rvert \leq   \lvert \mathcal{I}(C^T) \cap B_{2\tau}^s(\mathcal{I}(c_R^T)) \rvert,
\end{equation}
for any $c_R \in C$.

In \cite[Theorem 6]{raviv_2016}, the result on the list decodability of Gabidulin codes are also applied to lifted Gabidulin codes, obtaining families of subspace codes that cannot be list decoded efficiently at any radius. 
Also, the results of this paper can be applied to the corresponding lifted code. Indeed, we have the following results.

\begin{theorem}
Let $h,k, n$ and $m \in \mathbb{Z}^+$ such that $h \leq k \leq n \leq m$. Let $\mathcal{C}$ be a rank-metric code of $\mathcal{L}_{m,\sigma}$ and let $\mathrm{C}$ be the evaluation code over $n$ arbitrary $\F_q$-linearly independent elements $\alpha_1,\dots,\alpha_n\in \F_{q^m}$. 
Let $\tau$ be an integer such that $\left\lfloor\frac{d-1}{2}\right\rfloor+1 \leq \lfloor \tau/2 \rfloor \leq d-1$.
Suppose that $j < \tau$ and $\mathcal{C}$ contains $(\mathcal{G}_{m,h,\sigma})^{\sigma^j}$, for some $j< m-h$. Denote by $\mathcal{I}(\mathrm{C}^T)$ the $(n+m,M_s,2d,n)$ subspace code from the lifting of the code $\mathrm{C}^T$. Then there exists a subspace $\mathcal{I}(\mathbf{w}^T) \in \mathcal{G}_q(n+m,n)$, where $\mathbf{w} \in \F_{q^m}^n \setminus \mathrm{C}$ such that 
$$
\lvert \mathcal{I}(\mathrm{C}^T) \cap B_{2(n-l)}(\mathcal{I}(\mathbf{w}^T)) \rvert \geq \frac{{n \brack n-\lfloor \tau/2 \rfloor}_q}{q^{m(n-\lfloor \tau/2 \rfloor-h)}}.$$
\end{theorem}

Moreover, as a consequence of Theorem \ref{teo:main} and Equation \eqref{eq:liftball} we get the following result.

\begin{theorem} \label{teo:mainlifted}
Let $n, m \in \Z$ be positive integers such that
$n \mid m$. Let $\mathcal{C}$ be a rank-metric code of $\mathcal{L}_{m,\sigma}$ and let $\mathrm{C}$ be the associated evaluation code over an $\F_q$-basis of $\beta \F_{q^n}$, for some $\beta \in \F_{q^m}^*$, with minimum distance $d$. 
Let $l$ a positive integer such that $n-d+1\leq l\leq n-\lfloor \frac{d-1}{2} \rfloor-1$
and suppose that $(\mathcal{G}_{n,h,\sigma})^{\sigma^j} \subseteq \mathcal{C}$, for some $j \in \Z^+$ with $j \leq n-l$ and $l \geq h$.

Suppose that exists a subset $Sub \subseteq \mathrm{Pol}_{l,h}$ of $\sigma$-subspace polynomials of size $g$. Denote by $\mathcal{I}(\mathrm{C}^T)$ the $(n+m,M_s,2d,n)$ subspace code from the lifting of the code $\mathrm{C}^T$. Then there exists a subspace $\mathcal{I}(\mathbf{w}^T) \in \mathcal{G}_q(n+m,n)$, where $\mathbf{w} \in \F_{q^m}^n \setminus \mathrm{C}$ such that 
$$
\lvert \mathcal{I}(\mathrm{C}^T) \cap B_{2(n-l)}(\mathcal{I}(\mathbf{w}^T)) \rvert \geq g.
$$
\end{theorem}

Actually, we can rephrase all the results of the previous sections for the corresponding lifted code.
For instance, we have also rephrased the results regarding the lifting of the codes of Form \eqref{eq:codes}.

\begin{theorem}\label{thm:listdecGenatallrefnew}
Let $k,t,n$ and $m \in \mathbb{Z}^+$ such that $k \leq n$, $\lfloor t/2 \rfloor \mid n$ and $n \mid m$. Let $\mathcal{C}=\mathcal{H}_{n,k,\sigma}(f_1,f_2)$ as in \eqref{eq:codes} and let $\mathrm{C} = \mathrm{H}_{n,k,\sigma}(f_1,f_2)$ be the associated evaluation code over an $\F_q$-basis of $\beta \F_{q^n}$, for some $\beta \in \F_{q^m}^*$, where $f_2$ is not the zero polynomial.
Suppose $d(\mathrm{C})=d=n-k+1$, $\lfloor t/2 \rfloor \leq  d-1$ and 
\[ \lfloor t/2 \rfloor \geq \left\{ \begin{array}{ll} \left\lfloor \frac{d-1}{2} \right\rfloor+1 & \text{if}\,\, n-k\,\,\text{is even},\\  \\
\left\lfloor \frac{d-1}{2} \right\rfloor+2 & \text{if}\,\, n-k\,\,\text{is odd,}\\ 
\end{array} \right. \]
Denote by $\mathcal{I}(\mathrm{C}^T)$ the $(n+m,q^{kn},2d,n)$ subspace code from the lifting of the code $\mathrm{C}^T$. Then 
 \begin{itemize}
    \item The code $\mathrm{C}$ cannot be list decoded efficiently at any radius $\tau \geq t$.
    \item  If $d=2 \lfloor t/2 \rfloor -2$, then $\mathcal{I}(\mathrm{C}^T)$ cannot be list decoded efficiently at any radius $\tau$ greater than or equal to $\left\lfloor \frac{d-1}2 \right\rfloor+2$.
    \item  If $d=2\lfloor t/2 \rfloor -1$, then $\mathcal{I}(\mathrm{C}^T)$ cannot be list decoded efficiently at all.
\end{itemize}
\end{theorem}

Also, we mention the following consequence, which extends \cite[Theorem 7]{raviv_2016}.

\begin{theorem} \label{teo:usesubfield}
Let $p$ be a prime, $q$ be a power of $p$ and $k,t$ be positive integers. Consider $q'=p^r$, where $r$ is a non-negative integer, let $p_i=1+q'+q'^2\ldots+q'^i$, with $i \geq 0$, let $n=tp_k$ and $m$ be a multiple of $n$. 
Let $s$ be a positive integer coprime with $n$ and let $\sigma\colon x \in \fqn \mapsto x^{q^s}\in \fqn$.
Let $\mathcal{C}$ be a rank-metric code of $\mathcal{L}_{m,\sigma}$. Let $\mathrm{C}$ be the evaluation code of $\mathcal{C}$ over an $\F_q$-basis of $\beta\F_{q^n}$, for some $\beta \in \F_{q^m}^*$, and minimum distance $d$. Suppose that there exists a positive integer $t \geq 1$ such that:
\begin{itemize}
    \item[1.] $\lfloor \frac{d-1}{2} \rfloor+1 \leq tq'^k \leq d-1$;
    \item[3.] $(\mathcal{G}_{n,tp_{k-2}+1,\sigma})^{\sigma^j} \subseteq \mathcal{C}$, with $j< tq'^k$. 
\end{itemize}
Denote by $\mathcal{I}(\mathrm{C}^T)$ the $(n+m,c,2d,n)$ subspace code from the lifting of the code $\mathrm{C}^T$, where $c$ is the size of $\mathrm{C}$. Then 
The code $\mathrm{C}$ cannot be list decoded efficiently at any radius $\tau \geq tq^{'k}$.
\end{theorem}

\section{Conclusions and open problems}\label{sec:probl}

Applying the techniques of \cite{wachter-zhe_2013} and of \cite{raviv_2016,TrombZullo}, we obtain a Johnson-like bound for rank-metric codes containing generalized Gabidulin codes and we give a bound on the list size of rank-metric codes containing generalized Gabidulin codes, which relies on the existence of an enough large family of $\sigma$-subspace polynomials. 
These results allow us to exhibit families of rank-metric codes which are not efficiently list decoded either at all or from a larger number of errors, some of them are also MRD-codes or close to be MRD.
Our results suggest that rank-metric codes which are $\F_{q^m}$-linear or that contains a (power of) generalized Gabidulin code cannot be efficiently list decoded for large values of the radius (at least when $n\mid m$).
This is also supported by the fact that the known constructions of rank-metric codes (see e.g.\ \cite{guruswami_2016}) which can be efficiently list decoded are far from being $\F_{q^m}$-linear. 
A further important problem regards the divisibility conditions between $n$ and $m$. Do the results obtained in this paper and those in \cite{wachter-zhe_2013,raviv_2016,TrombZullo} hold true also when $n \nmid m$?
\cite[Lemma 7 and Corollary 3]{wachter-zhe_2013} imply a positive answer to the above question, but a general approach is still missing.

\section{Acknowledgements} 

This research was supported by the Italian National Group for Algebraic and Geometric Structures and their Applications (GNSAGA - INdAM).
The second author is also supported by the project ``VALERE: VAnviteLli pEr la RicErca" of the University of Campania ``Luigi Vanvitelli''.

\bigskip

\noindent Paolo Santonastaso and Ferdinando Zullo\\
Dipartimento di Matematica e Fisica,\\
Universit\`a degli Studi della Campania ``Luigi Vanvitelli'',\\
Viale Lincoln, 5\\
I--\,81100 Caserta, Italy\\
{{\em \{paolo.santonastaso,ferdinando.zullo\}@unicampania.it}}

\end{document}